\documentclass[11pt, letterpaper]{amsart}
\usepackage{graphicx, amssymb, color}
\usepackage{amsmath}

\newcommand\E{\ensuremath{\mathbb{E}}}
\newcommand\R{\ensuremath{\mathbb{R}}}

\newcommand\PP{\ensuremath{\mathbb{P}}}

\newcommand\la{\lambda}
\newcommand\La{\Lambda}
\newcommand\bW{\mathbf{W}}
\newcommand\eps{\varepsilon}

\newtheorem{thm}{Theorem}[section]
\newtheorem{cor}[thm]{Corollary}
\newtheorem{lemma}[thm]{Lemma}
\newtheorem{prop}[thm]{Proposition}

\theoremstyle{remark}
\newtheorem{rem}{Remark}[section]

 \numberwithin{equation}{section}

\title{Liquidation in Limit Order Books with Controlled Intensity}

\author{Erhan Bayraktar }\thanks{E. Bayraktar is supported in part by the National Science Foundation under an applied mathematics research grant and a Career grant, DMS-0906257 and DMS-0955463, respectively, and in part by the Susan M. Smith Professorship.}

\address[E. Bayraktar]{Department of
  Mathematics, University of Michigan, Ann Arbor, MI 48109}
\email{erhan@umich.edu}

\author{Michael Ludkovski}
\address[M.\ Ludkovski]{Department of Statistics and Applied Probability, University of California Santa Barbara, CA 93106-3110}
\email{ludkovski@pstat.ucsb.edu}


\begin{document}
\maketitle

\begin{abstract}

We consider a framework for solving optimal liquidation problems in limit order books. In particular, order arrivals are modeled as a point process whose intensity depends on the liquidation price. We set up a stochastic control problem in which the goal is to maximize the expected revenue from liquidating the entire position held. We solve this optimal liquidation problem for power-law and exponential-decay order book models explicitly and discuss several extensions. We also consider the continuous selling (or fluid) limit when the trading units are ever smaller and the intensity is ever larger. This limit provides an analytical approximation to the value function and the optimal solution. Using techniques from viscosity solutions we show that the discrete state problem and its optimal solution converge to the corresponding quantities in the continuous selling limit uniformly on compacts.
\end{abstract}

\section{Introduction}
Liquidation of large securities positions has emerged as an important problem in financial mathematics, linking together models of market microstructure and control theory. In this paper we consider an investor who liquidates a position through limit orders placed in a limit order book (LOB). The investor does so by choosing the price of the limit order; the higher the price of the limit order, the smaller the probability that it would be filled. The objective of the investor is to come up with an optimal limit order strategy that maximizes her expected revenue by date $T$.

Our model for the above problem is based on a point-process view of limit order books which treats liquidation as a sequence of discrete events, i.e.~order matches. More precisely, we assume that the investor effectively controls the frequency of her trades by choosing the spread $s$ above the current bid price $P_t$. The trade intensity is controlled as $\Lambda(s)$ and when a trade occurs, the investor generates a liquidation profit of $s$. Similar setups have been proposed in \cite{sa08}, \cite{ContStoikovTalreja,ContDeLarrard10} and rely essentially on a queueing system representation of LOB's.

A crucial modeling difference is whether execution takes place through market or limit orders. If investor trades via market orders, she necessarily encounters price impact through ``eating away'' a portion of the LOB. The precise price impact depends on the \emph{shape} of the LOB, as well as its resilience. Conversely, there is no transactions or fill risk as market orders execute instantaneously. This point of view is taken in e.g., \cite{ObizhaevaWang06,schied07,AlfonsiSchied10}. On the other hand, if the investor trades through limit orders, liquidation depends on being ``lifted'' by a sufficiently large market order, leading to substantial fill risk that again depends on the shape and depth of the LOB. The fill risk is related to the concept of virtual price impact \cite{WeberRosenow05} and is the focus of our model here. Related approaches to trading via limit orders can be found in \cite{sa08,CarteaJaimungal10,GueantLehalle11,GuilbaudPham11}. Also, in  our previous work \cite{BL10} we considered the same LOB as here but with an uncontrolled trade intensity and temporary price impact from order size.

Ideally, a fully specified model will reconcile the two approaches above, as well as consider the underlying risk preferences of the investor. This is especially important for dealing with simultaneous trading on multiple exchanges, see the very recent preprints \cite{KratzSchoneborn10,SchiedKlock11}, as well as Section \ref{sec:multi-scale} below. A full treatment of this problem will be the subject of a separate paper.

Our starting point is a discrete-state problem for an investor holding $n$ shares of an illiquid asset, $n \in \mathbb{N}$. Since practically speaking $n$ is often large (on the order of hundreds of thousands), we also investigate the fluid limit of our setup. On a technical level, the fluid limit provides asymptotic results for the discrete-state problem (see Remark~\ref{rem:asyp}), which is the main focus of our paper.

On a formal level, our control problem is equivalent to a controlled death process and is closely related to fluid approximations of some queueing problems. We refer to \cite{BauerleAAP00,Bauerle01,BauerleAAP02,day10,PiunovskiyMMOR09,Piunovskiy11} and references therein for the most relevant strand of this rich literature. In contrast with the previous literature, which uses probabilistic arguments, we utilize viscosity techniques to show convergence (both of the value functions and the corresponding optimal controls) from the discrete- to the continuous-state problems.


An investor holds $n$ shares of an asset. Let $(P_t)_{t \geq 0}$ be the bid price process for the underlying asset. Let $r \ge 0$ be the risk-free rate. We assume that $e^{-r t} P_t$ is a martingale with respect to the optimization measure $\PP$ on a filtered probability space $(\Omega, \mathcal{F}, (\mathcal{G}_t))$. This assumption is consistent with standard market microstructure models, see e.g.~\cite{AlfonsiSchied10}. Let $\Lambda_t$ be the (controlled) intensity of order fill, and let $s_t \ge 0$ be the spread between the bid price and the limit order of the investor. Denote by $N_t$ the $\mathcal{G}$-adapted counting process of order fills and $\tau_k$ the corresponding arrival times,
$$ N_t = \sum_k 1_{\{\tau_k \le t\}}.$$
Then $N_t - \int_0^t \Lambda_s \,ds$ is a martingale and expected revenue is
\begin{align}\label{eq:disc-revenue}
\E\left[ \sum_{i=1}^n e^{-r \tau_i} (P_{\tau_i} + s_{\tau_i}1_{\{\tau_i \leq T\}})\right].
\end{align}
We assume that the investor has a deadline date $T \le +\infty$ by which all trades must be completed. Remaining shares are liquidated at zero profit at $T$.

To introduce the liquidation control, we assume that $\Lambda_t = \Lambda(s_t)$, so that the intensity of order fills is a function of the offered spread above the bid price. Moreover, we assume that the bid price $P$ is unaffected by the limit orders created via $(s_t)$. Since $e^{-r t}P_t$ is a martingale, the first term in \eqref{eq:disc-revenue} is independent of $\tau_i$. Indeed, $\E[ \sum_{i=1}^n e^{-r \tau_i} P_{\tau_i}] = n P_0$ and we may ignore $P$ in the subsequent analysis.

We define
\begin{align}\label{defn:V}
V(n,T) & := \sup_{(s_t) \in \mathcal{S}_T} \E \left[ \sum_{i=1}^n e^{-r \tau_i} s_{\tau_i}1_{\{\tau_i \leq T\}} \right] \\ &=  \sup_{(s_t) \in \mathcal{S}_T} \E \left[\int_0^{T \wedge \tau(X)}e^{-rt}s_t  \, dN_t\right]  = \sup_{(s_t) \in \mathcal{S}_T} \E \left[\int_0^{T \wedge \tau(X)}e^{-rt}s_t \Lambda(s_t) \, dt\right], \label{defn:V2}
\end{align}
where
\[
\tau(X):=\inf\{t \geq 0: X_t=0\}\]
is the time of liquidation.
Here, $X_t:=X_0-N_t$, with $X_0=n$, is a ``death" (or inventory) process with intensity $\Lambda(s_t)$.
Note that $T$ in \eqref{defn:V} represents \emph{time-to-maturity} and $ \mathcal{S}_T$ is the collection  of $\mathcal{F}$-adapted controls, $s_t \ge 0$ with $\mathcal{F}_t := \sigma (N_s : s \le t)$.  The boundary conditions on $V$ are $V(n,0) = 0\;\forall n$ (terminal condition in time) and $V(0,T) = 0 \;\forall \;T$ (exhaustion).

\begin{rem}
Our model is related to the limit order book setup of \cite{sa08}, which assumes that limit orders are ``lifted'' through sufficiently large \emph{market} buy orders. Namely, a market buy order of size $q$, hits all limit sell orders that are within $I(q)$ of the best bid. Assuming that buy market orders arrive in the form of a Poisson random measure on $\R_+ \times \R_+$ with arrival intensity $\bar{\la} dt$ and volume (mark) distribution $f(dq)$, $q \ge 0$, a sell limit order at a given spread $u$ is lifted with probability $\PP( I(q) > u)$. By the thinning lemma on Poisson processes, such matching buy orders form a Poisson process with intensity $\bar{\la} \int_{I^{-1}(u)}^\infty f(dq)$. Empirical studies suggest a power-law depth function $f(dq) \propto q^{-1-a} dq$ \cite{sa08} and therefore if $\Lambda(0) < \infty$, we can view our model within the \cite{sa08} framework,  $\Lambda(s) \propto [I^{-1}( s)]^{-a}$, with $I^{-1}$ the virtual price impact function \cite{WeberRosenow05}.
\end{rem}

As in \cite{MR2519845}, the above control problem can be transformed into a discrete-time Markov decision problem and the classical results from \cite[Ch. 8]{MR511544} can be used to prove a dynamic programming principle. Using the latter result one can show that the value function is a viscosity solution of
\begin{align}\label{eq:hjb-V}
-V_T + \sup_{s \ge 0} \Lambda(s) \bigl[ V(n-1,T) - V(n,T) + s \bigr] - rV(n,T) = 0,
\end{align}
with boundary conditions $V(0,T)=V(n,0)=0$ and $V_T$ denoting partial derivative with respect to time-to-expiration. Standard results also imply that an optimal control can be taken of Markov feedback type, $s^*_t = s(X^*_t, T-t)$.
However,
in most of the examples below we will obtain explicit solutions to this dynamic programming equation. Then a \emph{verification lemma} can be used to justify that the solution of \eqref{eq:hjb-V} is indeed the value function.

The optimization problem described in \eqref{defn:V} is simplified but highly tractable. In most of the examples below, we are able to obtain closed-form solutions which provide direct insight into the relationship between the LOB model and its depth function and the investor's liquidation strategy. In Section \ref{sec:power-law} we give an explicit solution for \eqref{defn:V} in the case of a power-law intensity control $\La(s)$. Section \ref{sec:fluid} then studies convergence of the discrete problem \eqref{eq:hjb-V} to its continuous-state fluid limit. Our key Theorem \ref{thm:fir}, complemented by Proposition \ref{prop:VDinctov} and Corollary \ref{cor:fc}, gives a full account of this convergence using techniques from viscosity solutions of nonlinear partial differential equations. In Section \ref{sec:exp-law} we return to \eqref{defn:V} for the case where $\La(s)$ is of exponential shape; we are again able to provide several closed-form solutions. Finally, Section \ref{sec:extensions} considers several extensions and numerical illustrations of \eqref{defn:V}, including generic $\La(s)$, which shed additional light on the problem structure.

\section{Power-Law Limit Order Books}\label{sec:power-law}
In this section we assume that incoming buy orders have a power-law distribution for the spread, $\Lambda(s) = \frac{\lambda}{s^\alpha}$ for some $\alpha > 1$. It can be observed from the computations below that if $\alpha \le 1$, then no optimal control exists. Similar assumption was made (and justified empirically) by \cite{sa08} who write that in realistic markets $\alpha \in [1.5,3]$.

\begin{prop}\label{Prop:1}
Assume that $\Lambda(s) = \lambda s^{-\alpha}$ with boundary conditions $ V(0,T)=V(n,0) =0$ for all $n$. Then the solution of \eqref{eq:hjb-V} and the optimal spread are respectively
\begin{align}\label{eq:main-model}
V(n,T) = c_n \left(1-e^{-r \alpha T}\right)^{1/\alpha}, \qquad s^*(n,T) = \left( \frac{ \la}{\alpha r  c_n} \right)^{1/(\alpha-1)} \cdot \left(1-e^{-r \alpha T}\right)^{1/\alpha},
\end{align}
with $c_n$ satisfying the recursion
\begin{align}\label{eq:power-c-recursion}
r c_n = A_\alpha \lambda (c_n - c_{n-1})^{1-\alpha}, \quad n \ge 1, \qquad\qquad c_0 = 0,
\end{align}
where
\begin{align}
A_\alpha:= \frac{(\alpha-1)^{\alpha-1}}{\alpha^\alpha}.
\end{align}
\end{prop}

\begin{rem}\label{rem:cnv}
Note that $V$ is ``concave" in $n$ in the sense that
\[
V(n+1,T)-V(n,T) \leq V(n,T)-V(n-1,T),
\]
i.e., its linear interpolation in $n$ is concave in the usual sense.
This follows immediately from \eqref{eq:power-c-recursion} since
\[
c_{n+1}-c_n = \left(\frac{r c_{n+1}}{\lambda A_{\alpha}}\right)^{1/(1-\alpha)} \leq \left(\frac{r c_{n}}{\lambda A_{\alpha}}\right)^{1/(1-\alpha)} = c_{n}-c_{n-1},
\]
where the inequality follows from the fact that $c_n$ (or $V(n,T)$) is increasing in $n$. The latter follows directly from \eqref{defn:V}.

Also observe from \eqref{eq:main-model} and \eqref{eq:power-c-recursion} that
\[
s^*(n,T)=\frac{\alpha}{\alpha-1}(V(n,T)-V(n-1,T)),
\]
which implies that $n \mapsto s^{*}(n,T)$ is a decreasing function, because $n \mapsto V(n,T)$ is ``concave''. One can also think of $s^*$ as the derivative of the linear interpolation of $V$ in $n$.
\end{rem}

\begin{proof}
With the power law assumption \eqref{eq:hjb-V} reduces to
\begin{align}\label{eq:power-V}
-V_T + \sup_s \frac{\lambda}{ s^\alpha}( V(n-1,T)-V(n,T)+s) - rV = 0,
\end{align}
and therefore the candidate optimal policy is $t\mapsto s^*(X_t,T-t)$ in which $s^*(n,T) = \frac{\alpha}{\alpha-1} (V(n,T) - V(n-1,T))$.
To begin solving this equation, we start with $n=1$. Since $V(0,T) = 0$ for all $T$, we obtain for $V=V(1,T)$
$$
-V_T + A_\alpha \lambda V^{1-\alpha} - rV = 0. $$
This is a separable ordinary differential equation (ODE) which simplifies to
$$
T +C = \int^\cdot \frac{ V^{\alpha-1}}{A_\alpha \lambda - rV^{\alpha}} dV.
$$
Using the boundary condition $V(1,0) = 0$ we integrate to obtain
$$
\log( A_\alpha \lambda - rV^\alpha) = - r\alpha(T+C) \quad \Longleftrightarrow \quad V(1,T)= \left\{ \frac{A_\alpha \lambda}{r} (1-e^{-r \alpha T}) \right\}^{1/\alpha}.
$$
Considering the equation for general $n > 1$ we therefore make the ansatz $V(n,T) = c_n (1-e^{-r \alpha T})^{1/\alpha}$ and plugging into \eqref{eq:power-V} the relation \eqref{eq:power-c-recursion} follows.
\end{proof}

\begin{rem}
If no discounting is present $r=0$, one can verify that the solution of \eqref{eq:hjb-V} is $V(n,T) = d_n T^{1/\alpha}$, where the sequence $(d_n)$ satisfies the recursion
$$
d_n = \la \left( \frac{\alpha-1}{\alpha} \right)^{\alpha-1} (d_n - d_{n-1})^{1-\alpha}, \qquad d_0 = 0.
$$
This result can also be obtained by taking the limit $r \to 0$ in \eqref{eq:main-model}, \eqref{eq:power-c-recursion}.
\end{rem}

Fixing $X_0 = n$, the inter-trade intervals $\sigma_i := \tau_i  - \tau_{i-1}$, $i \le n$ have survival functions given by
$$
\PP( \sigma_i > t  | \tau_{i-1}) = \exp\left(-\int_0^t \La(s^*(n-i+1,T-\tau_{i-1}-s)) \,ds\right).
$$
Noting that $\La(s^*(n,T)) = \frac{C(n)}{1-e^{-r \alpha T}}$ for some constant $C(n)$, it follows that $\int_0^\eps \La(s^*(n,T)) \, dT = +\infty$ for all $n$ and $\eps$ and therefore $\PP( \sigma_i \le T -\tau_{i-1} ) = 1$ for all $i \leq n$. We conclude that even though there is no direct penalty if some orders remain at $T$,  with probability one, the full inventory is liquidated by $T$, $X^*_T = 0$ $\PP$-a.s. In particular, the problem with a \emph{hard} liquidation constraint $V(x,0) = -M 1_{\{x \ge 0\}}$ for any liquidation penalty $M \ge 0$ will have the same solution as in Proposition \ref{Prop:1}.

\subsection{Infinite Horizon}
As the execution horizon $T$ grows, the investor faces a weaker liquidation constraint. Nevertheless, she still prefers to sell earlier than later due to the discount parameter $r$ that incentivizes faster liquidation. For the limit $T\to\infty$ we obtain an infinite-horizon model whereby strategies are time-homogenous.

Taking $T \to \infty$ in \eqref{eq:main-model} we find that $V(n) = c_n$ and $s^*(n) = \la^{1/(\alpha-1)} (\alpha r c_n)^{1/(1-\alpha)}$. To understand how quickly execution takes place let us introduce \emph{expected time to liquidate} $S(n)$ which is defined to be $S(n) := \E[ \tau(X^*) | X^*_0=n]$,
in which $X^*$ is the death process whose intensity at time $t$ is $\Lambda(s^*(X^*_t))$ (representing optimally controlled inventory at $t$).
When the inventory is $X^*_t = n$, liquidation occurs at rate $\La(s^*(n))$, so that the interval until the next trade has an Exponential distribution with mean $1/\La(s^*(n))$. It follows that
\begin{align}
S(n) = \sum_{j=1}^n \La(s^*(j))^{-1} = \la^{1/(\alpha-1)} \sum_{j=1}^n  (\alpha r c_j)^{-\alpha/(\alpha-1)}.
\end{align}

\section{Continuous selling limit}\label{sec:fluid}
To better understand the results of Proposition \ref{Prop:1} we consider a limiting continuous model.
Let us denote the number of shares initially held by $x$.

We first introduce a sequence of discrete control problems that converge to the continuous selling limit. For $0 < \Delta \le 1$, consider the problem where shares are sold at $\Delta$ increments and the intensity of order fills is $\Lambda^{\Delta}(s):=\Lambda(s)/\Delta$. We will denote by $X^{\Delta}$ the ``death" process with this intensity and decrements of size $\Delta$. Then the resulting value function
\begin{align}\label{defn:V-delta}
V^{\Delta}(x,T) := \sup_{(s_t) \in \mathcal{S}_T} \E \left[ \sum_{i=1}^{x/\Delta} e^{-r \tau_i} \Delta  \cdot  s_{\tau_i}1_{\{\tau_i \leq T\}} \right]=  \sup_{(s_t) \in \mathcal{S}_T} \E \left[\int_0^{T \wedge \tau(X^{\Delta})}e^{-rt}s_t \Lambda(s_t)\,dt\right],
\end{align}
$x \in \{0, \Delta, 2 \Delta, \cdots\}$, $T \in \mathbb{R}_+$ would satisfy
\begin{equation}\label{eq:disinn}
-V^{\Delta}_T + \sup_{s \geq 0} \frac{\lambda} {s^\alpha \Delta}( V^{\Delta}(x-\Delta ,T)-V^{\Delta}(x,T)+s \Delta) - rV^{\Delta} = 0
\end{equation}
in viscosity sense.

Let us consider the first order partial differential equation (PDE)
\begin{equation}\label{eq:ffopde}
-v_T+\sup_{s \geq 0} \lambda \frac{s-v_x}{s^{\alpha}}-rv=0,
\end{equation}
which can be written as
\[
-v_T+ A_{\alpha}\lambda v_x^{1-\alpha}-r v=0,
\]
with boundary conditions $v(x,0)=v(0,T)=0$. The solution of \eqref{eq:ffopde} has the following deterministic control representation
\begin{equation}\label{eq:det-val-func}
v(x,T)=\sup_{(s_t) \in \mathcal{S}_T} \int_0^{T \wedge \tau(X^{(0),x})}\frac{\lambda}{s_t^{\alpha-1}}e^{-rt} dt,
\end{equation}
where $dX^{(0),x}_t=-\lambda s_t^{-\alpha}\,dt$, $X^{(0),x}_0=x$.
In fact, the solution of \eqref{eq:ffopde} is explicitly given by
\begin{align}\label{eq:toy-pde-sol}
v(x,T)= \left(\frac{\lambda}{r \alpha}\right)^{1/\alpha}x^{(\alpha-1)/\alpha}\left(1-e^{-r \alpha T}\right)^{1/\alpha}.
\end{align}
We denote the optimizer in \eqref{eq:ffopde} by $s^{(0)}(x,T)$, which is explicitly given by
\[
s^{(0)}(x,T)= \left(\frac{\lambda}{\alpha r}\right)^{1/\alpha}\frac{1}{x^{1/\alpha}}\left(1-e^{-r \alpha T}\right)^{1/\alpha}.
\]

\begin{rem}\label{rem:trade-curve-sec3} Plugging the optimizer back into the dynamics for $X^{(0),x}$ we obtain that
\[
dX^{(0),x}_t=- \frac{\alpha r X_t^{(0),x}}{1-e^{-r \alpha (T-t)}}dt,
\]
which can be explicitly solved as
\begin{align}\label{eq:trade-curve-power-det}
X^{(0),x}_t=x \exp\left({-\int_0^{t}\frac{\alpha r}{1-e^{-\alpha r (T-u)}}du}\right).
\end{align}
\end{rem}

Let $x \in \mathbb{R}_+$ be fixed and let us consider all the collections $\{0,\Delta, 2 \Delta, \cdots\}$ of grids that contain $x$ as an element. In the next result, will show that as $\Delta \to 0 $ then $V^{\Delta}(x) \to v(x)$.
In fact, the next result shows that this convergence is uniform on compacts.
\begin{thm}\label{thm:fir}
As  $\Delta \to 0$, $V^{\Delta} \to v$ uniformly on compact sets.
\end{thm}

\begin{proof}

\noindent Let us consider the regularized stochastic control problem
\begin{align}\label{defn:V-delta-eps}
V^{\Delta,k}(x,T) := \sup_{(s_t) \in \mathcal{S}^{k}_T} \E \left[ \sum_{i=1}^{x/\Delta} e^{-r \tau_i} \Delta  \cdot  s_{\tau_i}1_{\{\tau_i \leq T\}} \right], \qquad x \in \{0, \Delta, 2 \Delta, \cdots\},
\end{align}
where $\mathcal{S}^{k}_T:=\{s \in \mathcal{S}_T: s_t \in [1/k,k]\}$, $k>1$.
Using a representation similar to the one in \eqref{defn:V-delta} and using the lower bound on the controls $s \in \mathcal{S}^{k}_T$, it can be seen that
\begin{equation}\label{eq:bnd}
V^{\Delta,k}(x,T) \leq \frac{\lambda}{r}k^{\alpha-1}.
\end{equation}

We will follow the arguments of \cite{MR1115933} in the proof of their Theorem 2.1 (also see Theorem 4.1 on page 334 of \cite{MR2179357}) to show that $V^{\Delta,k}$ converges uniformly on compacts to the unique viscosity solution of
\begin{equation}\label{eq:const-fopde}
-v^{k}_T+\sup_{s \in [1/k,k]} \lambda \frac{s-v^k_x}{s^{\alpha}}-rv^k=0, \quad v^{k}(x,0)=0.
\end{equation}

Let $\bar{v}^k$ and $\underline{v}^k$ be defined by:
\[
\begin{split}
\bar{v}^k(x,T)&:=\limsup_{\delta \to 0} \limsup_{\Delta \to 0} \sup \left\{V^{\Delta,k}(y,S):|x-y|+|T-S| \leq \delta, \;y \in\{0, \Delta, \cdots\}  \right\},
\\ \underline{v}^k(x,T)&:=\liminf_{\delta \to 0} \liminf_{\Delta \to 0} \inf \left\{V^{\Delta,k}(y,S):|x-y|+|T-S| \leq \delta, \;y \in\{0, \Delta, \cdots\} \right\}.
\end{split}
\]
By definition we have that
$\underline{v}^k \leq v^k \leq \bar{v}^k$ and that $\underline{v}^k$ is lower semi-continuous, and $\bar{v}^k$ is upper semi-continuous; see e.g.~Proposition 5.2.1 of \cite{MR1484411}. We will show that $\bar{v}^k$ is a subsolution and that $\underline{v}^k$ is a supersolution of \eqref{eq:const-fopde}. It follows from Theorem 5.4.20 in \cite{MR1484411} that a comparison result holds for this PDE (the compactness of the control space is required in order to apply this result). This comparison theorem would then imply that $\bar{v}^k\leq \underline{v}^k$. As a result, $v^k=\bar{v}^k=\underline{v}^k$ is the unique continuous viscosity solution of \eqref{eq:const-fopde}. This fact together with the way the functions $\bar{v}^k$ and  $\underline{v}^k$ are defined also imply the local uniform convergence of $V^{\Delta,k}$ to $v^k$. (For a similar argument see page 35 of \cite{MR1118699}.)

We now prove that $\bar{v}^k$ is a viscosity subsolution of \eqref{eq:const-fopde}; the fact that $\underline{v}^k$ is a viscosity supersolution follows similarly.
Let $(x_0,T_0)$ be a local maximum of $\bar{v}^k-\phi$ for some test function $\phi \in C^{1,1}$. Without loss of generality, we will assume that $(x_0,T_0)$ is a strict local maximum and that $\bar{v}^k(x_0,T_0)=\phi(x_0,T_0)$, and $\phi \geq 2 \frac{\lambda}{r}k^{\alpha-1}$ outside the ball B$(x_0,T_0;\mathfrak{r})$, where $\mathfrak{r}>0$ is chosen so that
$(x_0,T_0)$ is the maximum of $\bar{v}^k-\phi$ on $B(x_0,T_0;\mathfrak{r})$. Thanks to the choice of the test function outside this ball, $(x_0,T_0)$ is in fact a global maximum of the function $\bar{v}^k-\phi$ and it is attained on $B(x_0,T_0;\mathfrak{r})$. (This is where the uniform boundedness assumption in \eqref{eq:bnd} is used.)

Let $(x^{\Delta}, T^{\Delta}) \in \{0,\Delta, \cdots\} \times \R_+ $ be a point at which $V^{\Delta,k}-\phi$ attains its (global) maximum. It follows from the definition of $\bar{v}^k$ and the fact that $(x_0,T_0)$ is a strict global maximum of $\bar{v}^k-\phi$ that there exists a sequence $\Delta_n \to 0$ such that $(x^{\Delta_n}, T^{\Delta_n}) \to (x_0,T_0)$, $V^{\Delta_n,k}-\phi$ attains its global maximum at that point and $V^{\Delta_n,k}(x^{\Delta_n}, T^{\Delta_n}) \to \bar{v}^k(x_0,T_0)$. From the global maximality
\[
 V^{\Delta_n,k}(x,T) -V^{\Delta_n,k}(x^{\Delta_n} ,T^{\Delta_n})\leq \phi(x,T) -\phi(x^{\Delta_n} ,T^{\Delta_n}).
\]
Moreover, it can be argued as in \cite{MR2519845} using the discrete dynamic programming principle (see \cite{MR511544}) that $V^{\Delta_n,k}$ satisfies
\[
-V^{\Delta_n,k}_T + \sup_{s \in [1/k,k]} \frac{\lambda} {s^\alpha \Delta_n} \left( V^{\Delta_n,k}(x^{\Delta_n}-\Delta_n ,T^{\Delta_n})-V^{\Delta_n,k}(x^{\Delta_n} ,T^{\Delta_n})+s \Delta_n \right) - rV^{\Delta_n,k} = 0
\]
in the viscosity sense. Then
\[
-\phi_T + \sup_{s \in [1/k,k]} \frac{\lambda} {s^\alpha \Delta_n} \left( \phi(x^{\Delta_n}-\Delta_n ,T^{\Delta_n})-\phi(x^{\Delta_n} ,T^{\Delta_n})+s \Delta_n \right) - r\phi+r(\phi-V^{\Delta_n,k}) \geq 0.
\]
Taking the limit as $\Delta_n \to 0$ we obtain from this equation that
\[
-\phi_T (x_0,T_0) + \sup_{s \in [1/k,k]} \frac{\lambda} {s^\alpha}( s-\phi_x(x_0,T_0)) - r\phi (x_0,T_0) \geq  0,
\]
which proves the subsolution property of $\bar{v}^{k}$. Here, we exchange the limit in $\Delta_n$ and the supremum with respect to $s$ using Proposition 7.32 in \cite{MR511544} which we can apply thanks to the compactness of the control space.

It follows again from Theorem 5.4.20 in \cite{MR1484411} that the unique solution of \eqref{eq:const-fopde} is given by
\begin{equation}\label{eq:aux-stcp}
v^k(x,T)=\sup_{(s_t) \in \mathcal{S}_T} \int_0^{T \wedge \tau(X^{(k),x})}\frac{\lambda e^{-rt}}{((s_t \vee 1/k) \wedge k)^{\alpha-1}} dt,
\end{equation}
where $dX^{(k),x}_t=-\lambda/((s_t \vee 1/k) \wedge k)^{\alpha}dt$, $X^{(k),x}_0=x$.
We will show that $v^{k}$ converges pointwise to $v$:
\[
\begin{split}
\lim_{k \to \infty}v^{k}(x,T)&=\sup_{k}\sup_{(s_t) \in \mathcal{S}_T} \int_0^{T \wedge \tau(X^{(k),x})}\frac{\lambda e^{-rt}}{((s_t \vee 1/k) \wedge k)^{\alpha-1}} \, dt \\ & =\sup_{(s_t) \in \mathcal{S}_T} \sup_{k}\int_0^{T \wedge \tau(X^{(k),x})}\negmedspace \frac{\lambda e^{-rt}}{((s_t \vee 1/k) \wedge k)^{\alpha-1}} \, dt
\geq \sup_{(s_t) \in \mathcal{S}_T}\int_0^{T \wedge \tau(X)}\!\frac{\lambda e^{-rt}}{s_t ^{\alpha-1}} \, dt=v(x,T),
\end{split}
\]
where the inequality follows from the lower semi-continuity of the map $X \mapsto \tau(X)$; see Lemma 5 in \cite{day10}. On the other hand, since $v$ is a supersolution of \eqref{eq:const-fopde}, the comparison result Theorem 5.4.20 in \cite{MR1484411} implies that $v^k \leq v$ for each $k$, and as a result
\[
\lim_{k \to \infty}v^{k}(x,T) \leq v(x,T).
\]
Combining the last two inequalities, we obtain the pointwise convergence of $v^k$ to $v$.
Pointwise convergence, on the other hand, implies uniform convergence on compacts due to Dini's theorem, since we already know that $v$ is a continuous function of its arguments, and that $v^{k}$ is an increasing sequence of functions. The latter fact follows from the fact that $v^{k+1}$ is a supersolution of the PDE $v^k$ satisfies.
\end{proof}

\begin{rem}
Results somewhat similar to Theorem~\ref{thm:fir} appeared in \cite{BauerleAAP00,Bauerle01,BauerleAAP02,day10,PiunovskiyMMOR09,Piunovskiy11} which are on the optimal control of queueing networks. (Among these papers only \cite{day10} considered optimal time-to-empty queueing control problems.) To prove Theorem~\ref{thm:fir} we used a completely different approach than the above literature, which had relied on probabilistic arguments. Our approach relies in contrast on the analytical approximation ideas of \cite{MR1115933}. We see the prelimit control problem as the discretization (only in the space variable but not in the time variable) of the ``fluid limit'' first order non-linear PDE \eqref{eq:ffopde} and rely on convergence of the approximation schemes to the viscosity solutions of such  non-linear PDEs. This approach could be fruitful in general in proving ``fluid limit'' results associated to controlled queueing networks.
\end{rem}

The following is a strengthening of Theorem~\ref{thm:fir} which is an interesting result in its own right.
\begin{prop}\label{prop:VDinctov}
For any sequence $(\Delta_k)$ with $\Delta_k = \delta 2^{-k}$, we have $V^{\Delta_k}\uparrow v$ as $k \to \infty$.
\end{prop}

\begin{proof}
We show that for any $\Delta>0$, $V^{2 \Delta} \leq V^{\Delta}$.  Due to the factoring of $T$ and $x$ in Proposition \ref{Prop:1}, it suffices to establish this result on the infinite horizon where strategies are constant between trading times.

Fix $\eps>0$ and let $s^{2\Delta}$ be an $\eps$-optimal strategy for $V^{2 \Delta}$. This policy is defined over $x \in \{0,2 \Delta, 4 \Delta, \cdots\}$.
We will recursively construct a policy $s^\Delta$ over the domain $x \in \{0,\Delta, 2\Delta, \cdots\}$ that outperforms $s^{2\Delta}$. The dynamic programming principle implies that
\[
\begin{split}
V^{2\Delta}(2n \Delta) &\le \E \Bigl[ e^{-r \tau_1}[2 s^{2\Delta}(2n \Delta) \Delta + V^{2\Delta}((2n-2)\Delta)\Bigr] + \eps
\\&= \frac{\La(s^{2\Delta}(2n \Delta))}{\La(s^{2\Delta}(2n \Delta))+r} \left\{2 s^{2\Delta}(2n \Delta) \Delta + V^{2\Delta}((2n-2)\Delta) \right\} + \eps.
\end{split}
\]
Similarly, given the liquidation strategy $s^{\Delta}$ and corresponding trading times $\tilde{\tau}_i$, the resulting expected profits denoted as $\tilde{V}^\Delta(x)$, $x\in \{0,\Delta, 2\Delta,\cdots\}$ satisfy for $y=2n \Delta$,
\begin{align}\notag
&\tilde{V}^\Delta( y) = \E \left[ e^{-r \tilde{\tau}_1}  s^{\Delta}(y)\Delta + e^{-r \tilde{\tau}_2}\left\{ s^{\Delta}( y-\Delta) \Delta + \tilde{V}^\Delta( y-2\Delta) \right\}  \right] \\  \notag
& \qquad = \frac{ 2\La(s^{\Delta}(y))}{2\La(s^{\Delta}(y) )+r} \left\{ s^\Delta(y) \Delta + \frac{ 2\La(s^{\Delta}( y- \Delta) )}{2\La(s^{\Delta}( y-\Delta)) +r} [ s^\Delta(y-\Delta) \Delta + \tilde{V}^\Delta( y-2\Delta)] \right\}.
\end{align}
Given $s_{2n} \equiv s^{2\Delta}(2n \Delta)$ we prove below that there exists $u \in \R_+$, such that
\begin{align}\label{eq:prop-2-2-ineq}
\frac{ \La(s_{2n})}{\La(s_{2n})+r} \left( 2 s_{2n} \Delta + V \right) \le \frac{ 2\La(u)}{2\La(u) + r} \left\{ u \Delta + \frac{2\La(u)}{2\La(u)+r}( u \Delta + V) \right\} ,
\end{align}
for any $V \ge 0$. This would establish $V^\Delta(2n \Delta) \ge \tilde{V}^\Delta(2n \Delta) \ge V^{2\Delta}(2n\Delta) -\eps$ by induction on $n$ after setting $s^{\Delta}(2n \Delta) = s^{\Delta}((2n-1)\Delta) = u$. Since $\eps$ is arbitrary, the statement of the proposition would then follow. Note that in the above construction, the $\Delta$-investor trading in smaller increments and twice as much, uses the \emph{same} spread $u$ to trade when her inventory is $2n\Delta$ or $(2n-1)\Delta$.

Let $z^2 := \frac{ \La(s_{2n})}{\La(s_{2n})+r}$ and define $u$ implicitly through $\frac{2\La(u)}{2\La(u)+r} := z < 1$.
Solving for $s_{2n}$ and $u$ in terms of $z$ and using $\La(s) = \lambda s^{-\alpha}$ we obtain
\[
s_{2n} = \left(\frac{\lambda(1-z^2)}{r z^2} \right)^{\alpha^{-1}} > \left( \frac{2 \lambda (1-z)}{r z } \right)^{\alpha^{-1}} = u.
\]
Observe that by construction
\[
\frac{ \La(s_{2n})}{\La(s_{2n})+r} =  \frac{ 4\La(u)^2}{(2\La(u) + r)^2},
\]
so that the Laplace transform at $r$ of the duration to execute two trades by the $\Delta$-investor is equal to the Laplace transform at $r$ of the duration to execute one trade by the $2\Delta$-investor.  Using this fact, \eqref{eq:prop-2-2-ineq} is equivalent to
\begin{align*}
\frac{2 \La(s_{2n})}{\La(s_{2n})+r} s_{2n} &\le  \frac{ 2\La(u)}{2\La(u) + r} \left\{ u + \frac{2\La(u)}{2\La(u)+r} u \right\} \\
\Longleftrightarrow \qquad 2 z^2 s_{2n} &\le z(1+z) u.
\end{align*}
Since $\alpha > 1$ and the terms on both sides of the above inequality are positive, we may raise both sides to the $\alpha$-power and plug-in the expressions for $s_{2n}$ and $u$ to find
\begin{align*}
(2 z^2)^\alpha s_{2n}^\alpha - z^\alpha (1+z)^\alpha u^\alpha & = (2 z^2)^\alpha \frac{\lambda(1-z^2)}{r z^2} - z^\alpha (1+z)^\alpha  \frac{2 \lambda (1-z)}{r z } \\
& = 2 \lambda r^{-1} (1-z)(1+z)z^{\alpha-1}  [ (2z)^{\alpha-1} - (1+z)^{\alpha-1}] < 0,
\end{align*}
where the last inequality follows since $z<1$ and $\alpha >1$. This shows that \eqref{eq:prop-2-2-ineq}  holds and concludes the proof of the proposition.
\end{proof}

Next, we show that the strategies also converge thanks to the concavity of all the functions involved.

\begin{cor}\label{cor:fc}
Let us denote by $s^{(\Delta)}$ the pointwise optimizer in \eqref{eq:disinn}. Then we have that $s^{(\Delta)}(x,T) \to s^{(0)}(x,T)$. (Here, $x$ is fixed and we take the limit over the grids that pass through $x$.)
\end{cor}

\begin{proof}
On the one hand, as in Remark~\ref{rem:cnv} we can think of  $s^{(\Delta)}$ as something proportional to the left derivative of the linear interpolation of $V^{\Delta}$, denoted by $\hat{V}^{\Delta}$,  which is increasing and concave. On the other hand, $s^{(0)}$ is proportional to the derivative of the concave differentiable function $v$. By Theorem 24.5 on page 233 of \cite{MR1451876} it however follows that for any $x > 0$,
\[
\begin{split}
|D^{-}_{x}&\hat{V}^{\Delta}(x,T)-v_x(x,T)| \leq \varepsilon
\end{split}
\]
for small enough $\Delta$, where $D^{-}_{x}$ denotes the left derivative operator with respect to $x$.
\end{proof}

\begin{rem}\label{rem:asyp}
Theorem~\ref{thm:fir} tells us about the asymptotics of $c_n$ in \eqref{eq:power-c-recursion}:
\[
c_n \sim \left(\frac{\lambda}{r \alpha}\right)^{1/\alpha}n^{(\alpha-1)/\alpha} \quad \text{as $n \to \infty$}.
\]
Corollary~\ref{cor:fc} can be used to find out the marginal price asymptotics in Proposition~\ref{Prop:1}:
\begin{equation}\label{eq:mpa1}
s^*(n,T) \sim \left(\frac{\lambda}{\alpha r}\right)^{1/\alpha}\frac{1}{n^{1/\alpha}} \quad \text{as $n \to \infty$}.
\end{equation}
Clearly, the spread will go to zero as $n \to \infty$. But here we are able to obtain the rate of convergence to zero as  a function of the remaining inventory.
\end{rem}


For time to execution on infinite horizon we have for $\tau_1 = \inf\{ t: X_t \le x_1 \}$ and $S(x_1, x_2) := \E[ \tau_1  | X_0 = x_2]$ that
\begin{align}
S(x_1,x_2) = \int_{x_1}^{x_2} \frac{1}{\La(s^{(0)}(u))} \,du,
\end{align}
since intuitively when inventory is of size $u$, the expected time to liquidate an infinitesimal quantity $du$ is inversely proportional to the current trading rate $\La(s^{(0)}(u))$. Plugging in $s^{(0)}(u) = (\frac{\la}{\alpha r u})^{1/\alpha}$ we obtain $\La(s^{(0)}(u)) = \alpha r u$ or
$S(x_1, x_2) = \frac{1}{\alpha r} \log (\frac{ x_2}{x_1} )$ which shows that orders are filled in logarithmic time (as $x_1 \to 0$ the remainder is executed arbitrarily slow).

\section{Exponential-Decay Order Books}\label{sec:exp-law}
The power-law order book implies that trades can be made arbitrarily quickly as the spread goes to zero: $\lim_{s \to 0} \La(s) = +\infty$. Also, it gives a relatively good chance of executing trades deep in the book, i.e.~when $s$ is large. For less liquid markets, both of these features might not be realistic. Accordingly, we consider an exponential-decay LOB, with
\begin{align}\label{eq:exp-lambda}
\La(s) = \la e^{- \kappa s}, \quad \kappa > 0,
\end{align}
where $\kappa$ controls the exponential depth of the book and $\la = \La(0)$ is the order intensity at the bid price. The optimization problem for the spread is now of the form
$$
\sup_{s \ge 0} \, \la e^{-\kappa s}\bigl( V(n-1)-V(n) + s \bigr),
$$
which leads to the candidate optimizer $s^*(n) = \frac{1}{\kappa} + (V(n) - V({n-1}))$. We observe that $s^*$ is bounded away from zero so no trades are ever placed close to the bid.

\subsection{Finite Horizon}
With a finite horizon and no discounting we obtain the following closed-form solutions to the execution problem.
\begin{prop}
Consider again $V^\Delta(x,T)$ defined in \eqref{defn:V-delta} with boundary condition $V^\Delta(x,0) =V^{\Delta}(0,T)=0$, $r=0$ and $\Lambda(s)$ given in \eqref{eq:exp-lambda}. Then for $x=n \Delta$,
\begin{align}\label{eq:Vn-exp-closed-form}
V^\Delta(x ,T)= \frac{\Delta}{\kappa} \log \left( \sum_{j=0}^n \frac{1}{j!} \left(\frac{ \lambda T}{\Delta e}\right)^j \right),
\end{align}
and
\begin{align*}
s^*(n\Delta, T) = \frac{1}{\kappa} \left( 1 + \log \left( 1 +  \frac{\frac{ (\lambda T)^n}{ (\Delta e)^n n!}}{\sum_{j=0}^{n-1} \frac{(\lambda T)^j}{(\Delta e)^j j!} } \right) \right).
\end{align*}
As $\Delta \searrow 0$, $V^\Delta(n\Delta,T) \to v(x,T)$ uniformly on compacts, where $v(x,T)$ solves the nonlinear first order PDE
\begin{equation}\label{eq:exp-finh-pde}
v_T(x,T) = \frac{\lambda}{\kappa} e^{-1 - \kappa v_x(x,T)},
\end{equation}
with boundary conditions  $v(0,T)=v(x,0)=0$.
The solution to this PDE satisfies
\begin{equation}\label{eq:bounds}
\frac{x}{\kappa} \log\left(\frac{\lambda}{x}T\right) \leq v(x,T)  \leq \frac{\lambda}{\kappa e} T.
\end{equation}
\end{prop}

\begin{rem}
In the above notation $x$ is the number of shares, which is fixed across the problems. When we are taking the continuous liquidation limit, we let $\Delta \downarrow 0$, the size of trading units,  while taking the number of units $n$ as $n\Delta = x$, for $x$ constant.
\end{rem}
\begin{proof}
$V^\Delta (n\Delta, T)$ satisfies the HJB equation
\begin{align}\notag
-\partial_T V^\Delta(n\Delta,T) + \sup_{s \ge 0} \la \Delta^{-1} e^{-\kappa s}( V^\Delta( (n-1)\Delta,T)-V^\Delta(n,T) + s\Delta) = 0, \end{align}
which can be written as
\begin{equation}
 \partial_T V^\Delta(n\Delta,T)  =  \frac{\la}{\kappa \Delta} \exp \left( -1+\Delta^{-1}\kappa(V^\Delta({(n-1)\Delta},T)-V^\Delta(n,T)) \right).
\label{eq:Vn-exp-recursion}
\end{equation}
Letting $B := \frac{\la}{\kappa e} $, integrating and using $V^\Delta(\cdot,0)\equiv 0$ we find for $n=1$ that
\begin{align}
V^{\Delta}(\Delta,T) = \frac{\Delta}{\kappa} \log (1+ B \kappa \Delta^{-1} T ).
\end{align}
Iterating over $n$ the separable ODE for $V^\Delta(n,\cdot)$ in \eqref{eq:Vn-exp-recursion} we obtain \eqref{eq:Vn-exp-closed-form}.
The expression for the optimal spread follows from $s^*(n\Delta, T) = \frac{1}{\kappa} + \frac{V^\Delta(n\Delta)-V^\Delta((n-1)\Delta)}{\Delta}$.

The proof that  as $\Delta \downarrow 0$, $V^\Delta(n\Delta,T) \to v(x,T)$ uniformly on compacts
can be proven as in Theorem~\ref{thm:fir}.

Observe that both bounds in \eqref{eq:bounds} satisfy \eqref{eq:exp-finh-pde}.
To prove the lower bound, let us introduce the following function:
\[
\tilde{V}^\Delta(x ,T)= \frac{\Delta}{\kappa} \log \left(\frac{1}{n!} \left(\frac{ \lambda T}{\Delta e}\right)^n \right).
\]
Clearly, $\tilde{V}^{\Delta} \leq V^{\Delta}$. From Stirling's formula we know that
\[
n! \sim \sqrt{2 \pi n} \left(\frac{n}{e}\right)^n,
\]
where we use $\sim$ to indicate that the ratio of the left to the right-hand-side converges to $1$ as $n \to \infty$. As a result, $\tilde{V}^{\Delta}(x,T) \sim x/\kappa \log(\lambda T/x)$, recalling that $x=n \Delta$. Now the lower bound in \eqref{eq:bounds} follows since $v(x,T) \geq \lim_{\Delta \downarrow 0}\tilde{V}^{\Delta}$. We could have provided an alternative proof using a comparison theorem for the first-order non-linear PDE  \eqref{eq:exp-finh-pde} since in fact  $x/\kappa \log(\lambda T/x)$ also satisfies this PDE with a smaller boundary value at $T=0$. We preferred to be more constructive in our proof.

The fact that $\lambda T/ (\kappa e)$ is an upper bound on $v$ follows directly from the observation that $v_T \leq \lambda/ (\kappa e)$ (recalling that $v$ is increasing in $x$) and that $v(x,0)=0$.

\end{proof}

\begin{rem}
Since the trading rate is bounded $\La(s^*) \le \la e^{-1}$, for $x > \la e^{-1} T$ the full inventory cannot be liquidated by horizon $T$. Therefore, in the region $\mathcal{D} := \{ (x,T) : x > \la e^{-1} T\}$, $v$ is independent of $x$ and the upper bound is tight: $v(x,T) = \frac{\la}{\kappa e} T$ on $\mathcal{D}$.
\end{rem}

\begin{rem}\label{rem:exp-exec-curve}
Here, we will determine the shape of $t \to X^{(0),x}_t$ in Remark~\ref{rem:trade-curve-sec3} for the exponential order books. First,
\begin{equation}\label{eq:dyoX}
d X^{(0),x}_t=-\Lambda(s(X^{(0),x}_t,T-t))dt,
\end{equation}
where $\Lambda$ is given by \eqref{eq:exp-lambda}. On the other hand, $s(X^{(0),x}_t,T-t)=\frac{1}{\kappa}+v_x(X^{(0),x}_t,T-t)$. Using this relationship, along with \eqref{eq:exp-finh-pde}, which implies that $t \mapsto s(X^{(0),x}_t,T-t)$ is a constant function (let us denote that value by $s^*$), it follows from \eqref{eq:dyoX} that
\[
\frac{d^2 X^{(0),x}_t}{dt^2}=0, \quad t \in [0,T],
\]
i.e., $t \mapsto X^{(0),x}_t$, $t \in [0,T]$, is a strictly decreasing linear function. In fact, one can compute $s^*$ by maximizing the value function \eqref{eq:det-val-func} (after replacing power rate with exponential) over constant spreads (since the optimal spread is known to be a constant). This yields that $s^*=\frac{1}{\kappa}\log\left(\frac{\lambda T}{x}\right)$ if $\lambda T/x \geq e$. Otherwise $s^{*}=1/\kappa$. The expression for the optimal spread and \eqref{eq:dyoX} in turn imply that if $x \leq \la e^{-1} T$,
$X^{(0),x}_t=x(1-\frac{t}{T})$ (which is 0 at $T$) and if $x > \la e^{-1} T$ we have that $X^{(0),x}_t=x-\lambda t/e$ (which remains strictly positive at $T$).
\end{rem}

\subsection{Infinite Horizon}
We also have closed-form expressions for the infinite horizon case.

\begin{prop}
For exponential-decay LOB with $T=+\infty$ and discounting rate $r >0$ we have
\begin{align}\label{eq:lambertw-inf-horizon}
V^\Delta(x) = \frac{\Delta}{\kappa} \bW\left( \la r^{-1} \Delta^{-1} \exp \Bigl(\kappa \frac{V^\Delta(x-\Delta)}{\Delta}-1 \Bigr) \right), \; x \in \{0, \Delta, \cdots\}, \qquad V^\Delta(0) = 0,
\end{align}
where $\bW$ is the Lambert-W function (or the double-log function defined as $z = \bW(y)$ for $z e^z = y$).

As $\Delta \downarrow 0$, $V^\Delta(x) \to v(x)$ uniformly on compacts where
\begin{align}\label{eq:ei-inf-horizon}
li\left( \frac{ e\kappa r v(x)}{\la} \right) & = -\frac{ e r x}{\la},
\end{align}
and
$li(y) := \int_{0}^y \frac{1}{\log t} \, dt$ is the logarithmic integral function.
\end{prop}

\begin{proof}
The HJB equation for $V^\Delta(x)$ is
\begin{align*}
-r V^\Delta(x) + \sup_{s \ge 0} \la \Delta^{-1} e^{-\kappa s}\left( V^\Delta(x-\Delta)-V^\Delta(x) + s \Delta \right) & = 0.
\end{align*}
Using the optimizer $s^{(\Delta)} = \frac{1}{\kappa} + \frac{V^\Delta(x)-V^\Delta(x-\Delta)}{\Delta}$ we reduce to
\begin{align*} V^\Delta(x) & =  \frac{\la \Delta}{\kappa r} \exp \left( -1+\kappa\frac{V^\Delta({x-\Delta})-V^\Delta(x)}{\Delta} \right),
\end{align*}
which has closed-form solution given by \eqref{eq:lambertw-inf-horizon}. Arguments similar to  Theorem~\ref{thm:fir} imply that $V^\Delta \to v$ uniformly on compacts and the continuous inventory limit satisfies
\begin{align}
-r v(x) + \sup_{s \ge 0} \la e^{-\kappa s}(s - v'(x)) &= 0.
\end{align}
Solving for $v'$ we obtain
\begin{align*}
v'(x) & = -\frac{1}{\kappa} \left\{1 + \log \left(\frac{\kappa r v(x)}{\la} \right)\right\}. \end{align*}
The last nonlinear first-order ordinary differential equation has closed-form solution given in \eqref{eq:ei-inf-horizon}.
Asymptotically $\lim_{x \to \infty} v(x) = \frac{\la}{\kappa r e}$ and the optimal spread is
\begin{align}\label{eq:s0-exp-law}
s^{(0)}(x) = \frac{1}{\kappa} \left\{\log \left( \frac{\la}{\kappa r v(x)} \right) \right\}.
\end{align}
\end{proof}

We note that since $v$ is increasing in $x$, $x \mapsto s^{(0)}(x)$ in \eqref{eq:s0-exp-law} is decreasing and so $v$ is concave. As before, $\lim_{x \to 0} s^{(0)}(x) = +\infty$ so the control space remains unbounded, however the pay-off rate $s^{(0)} \Lambda(s^{(0)})$ is bounded. Moreover, a direct check verifies that $V^\Delta$ and $v$ are inversely proportional to the exponential depth parameter $\kappa$, i.e.\ doubling $\kappa$ (making the order book more shallow) halves $V^\Delta$ and $v$, and correspondingly halves the optimal spreads $s^{(\Delta)}$ and $s^{(0)}$.

Figure \ref{fig:exp-power-s} graphically illustrates the difference between exponential-decay and power-law LOB's. As observed, for an exponential LOB, $s^*$ is bounded away from zero, while $\lim_{x\to\infty} s^*(x) = 0$ in power LOB's. Moreover, while $\lim_{x\downarrow 0} s^{(0)}(x) = +\infty$ in any LOB, the rate is much slower in an exponential LOB (due to thinner tail for large spreads) compared to power-LOB.

\begin{figure}[ht]
\center{\includegraphics[height=3in]{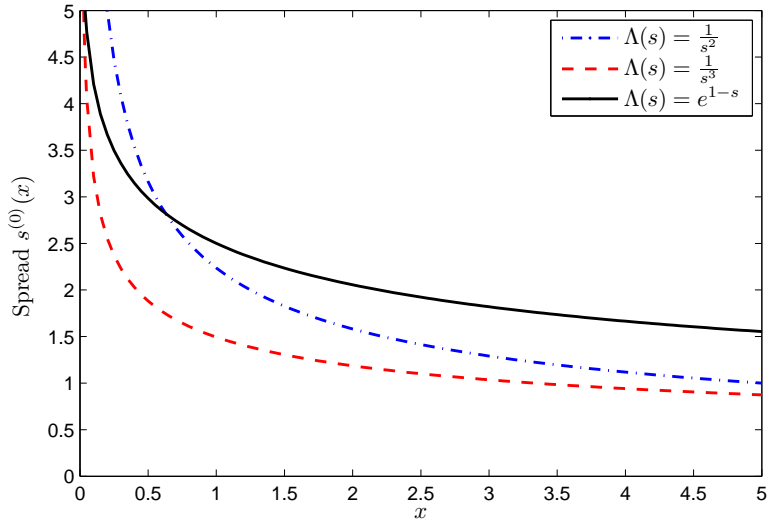}}
\caption{Optimal controls for power-law and exponential-decay order books. We take $r=0.1$ and depth functions $\La(s) \in \{ s^{-2}, s^{-3}, e^{1-s}\}$, which have been normalized such that $\La(1) = 1$ in all three cases. The plot shows
the resulting fluid limit spreads $s^{(0)}(x)$.  \label{fig:exp-power-s}}
\end{figure}

\section{Discussion and Further Extensions}\label{sec:extensions}
\subsection{Numerical Example: Convergence to the Fluid Limit}

To illustrate the convergence to the fluid limit consider the problem of selling up to $x=5$ blocks of shares in a power-law LOB with $\La(s) = s^{-2}$. We suppose that a block corresponds to 100 shares and that the minimal trading unit is either 5 or 1 shares, i.e.~$\Delta =0.05$ and $\Delta=0.01$ respectively. For a fixed $\Delta$, we can easily compute $V^\Delta(x=n\Delta)$ or $v(x)$ using the results in Section \ref{sec:power-law}. Figure \ref{fig:discrete-cont-v} illustrates the percent difference between $V^\Delta$ and the fluid limit $v$. As shown in Proposition \ref{prop:VDinctov}, $V^\Delta$ is decreasing in $\Delta$ and $\lim_{\Delta\downarrow 0} V^\Delta = v$. We observe that the convergence is quite rapid in $x$.

\begin{figure}[Ht]
\hspace*{-20pt}\begin{tabular}{cc}
\begin{minipage}{0.5\textwidth}
\center{\includegraphics[width=3.35in]{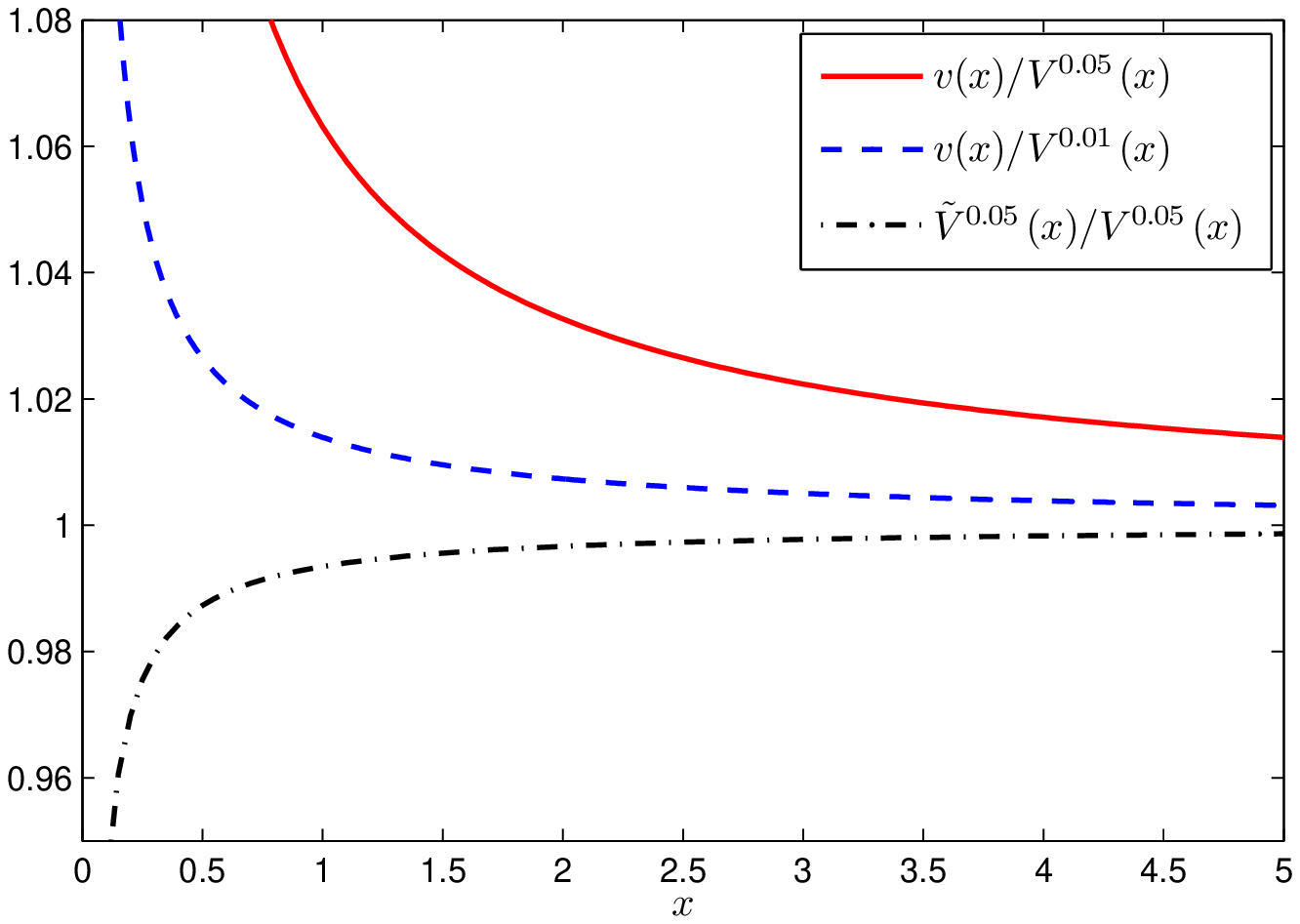}}
\end{minipage} & \begin{minipage}{0.5\textwidth}
\center{\includegraphics[width=3.35in]{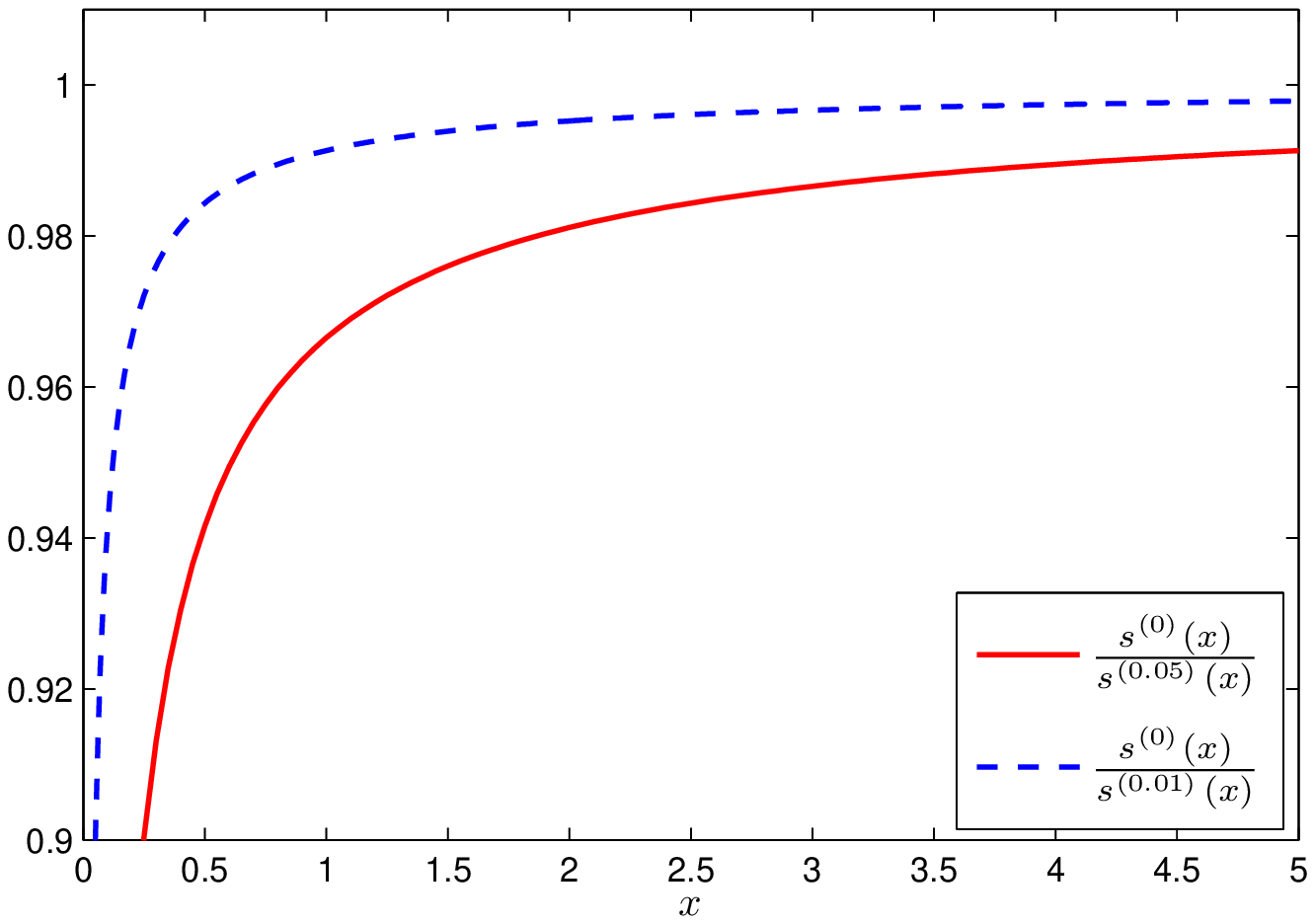}}
\end{minipage}\end{tabular}
\caption{Convergence to the fluid limit. Left panel: the ratio between discrete and continuous $V^\Delta(x)/v(x)$ for $\Delta =0.05$ and $\Delta=0.01$. Additionally, we plot $\tilde{V}^\Delta/V^\Delta$ as defined in \eqref{eq:cont-control-in-V-delta}. Right panel: ratio of the fluid limit optimal control $s^{(0)}(x)$ to the discrete $s^{(\Delta)}(x)$
for $\Delta\in \{0.01,0.05\}$.
 \label{fig:discrete-cont-v}}
\end{figure}

The right panel of Figure \ref{fig:discrete-cont-v} shows that the controls themselves are also very close. We observe that $s^{(\Delta)} \searrow s^{(0)}$. Note that here we are essentially comparing $s^{(0)}$ with its right-sided Riemann-sum approximation,  since $s^{(\Delta)}(x)$ corresponds to the spread charged for all shares in $[x,x-\Delta)$ while $s^{(0)}(x)$ corresponds to the marginal spread \emph{at} $x$. Accordingly, better approximations, such as $\check{s}^{(\Delta)}(x) := \frac{1}{\Delta} \int_{x-\Delta}^x s^{(0)}(u) \,du$, would make the discrete and fluid controls even closer.

Given the simple expression for the fluid limit control $s^{(0)}(x)$, a useful approximation is to use a discretized version of $s^{(0)}$ as an approximately optimal control for $V^\Delta$. Let
\begin{align}\label{eq:cont-control-in-V-delta}
\tilde{V}^\Delta(x) := \E\left[ \sum_{i=1}^{x/\Delta} e^{-r \tau_i} s^{(0)}( x-i\Delta) \right], \quad x \in\{0,\Delta, 2 \Delta, \ldots\},
\end{align}
represent the expected gains from a discrete strategy which uses a spread of $s^{(0)}( (n-i)\Delta)$ for the $i$-th trade of size $\Delta$. In the left panel of Figure \ref{fig:discrete-cont-v} we see that this approximation is excellent for $V^\Delta$ even for moderate values of $x$ (less than 1\% difference for $\Delta = 0.01$ and $x>1$).

\subsection{Execution Curves}
A popular way of describing a trade execution algorithm is through the execution curve $t \mapsto X_t/X_0$, see e.g.~\cite{Almgren00,Almgren03,GueantLehalle11}. In our model with execution risk, $X_t$ is a random variable, and we will therefore consider the natural analogue of average execution curve $E(x,t) :=  \E[ X^{*,x}_t]$, where $X^{*,x}$ is the remaining inventory at $t \le T$ starting with initial condition $X^{*,x}_0 = x$ . The baseline case where $\bar{E}(x,t) = x(1 - t/T)$ is linear, corresponds to ``linear price impact'' or zero-risk-aversion in \cite{Almgren00} and implies that the average trading rate is constant.

For notational convenience we temporarily fix $\Delta = 1$. Recall that
\[
dX^{*,x}_t = -\La( s^*(X^{*,x}_t, T-t)) \,dt + dM_t
\]
where $(M_t)$ is a martingale (the compensated order \emph{departure} process),
which implies by an application of It\^{o}'s formula that $(E(x,t))_{x=0}^\infty$ satisfies the system of inhomogenous linear ODEs
\begin{align*}
\frac{ dE(x,t)}{dt} =\La( s^*(x,T-t))( E({x-1},t) - E(x,t)), \quad E(x,0) = x,
\end{align*}
with $E(0,t) \equiv 0$.
%
%
Any finite collection of these ODEs can be solved analytically using integrating factors, or numerically with any standard solver.

\begin{figure}[h!]
\hspace*{-20pt}\begin{tabular}{cc}
\begin{minipage}{0.5\textwidth}
\center{\includegraphics[width=3.35in]{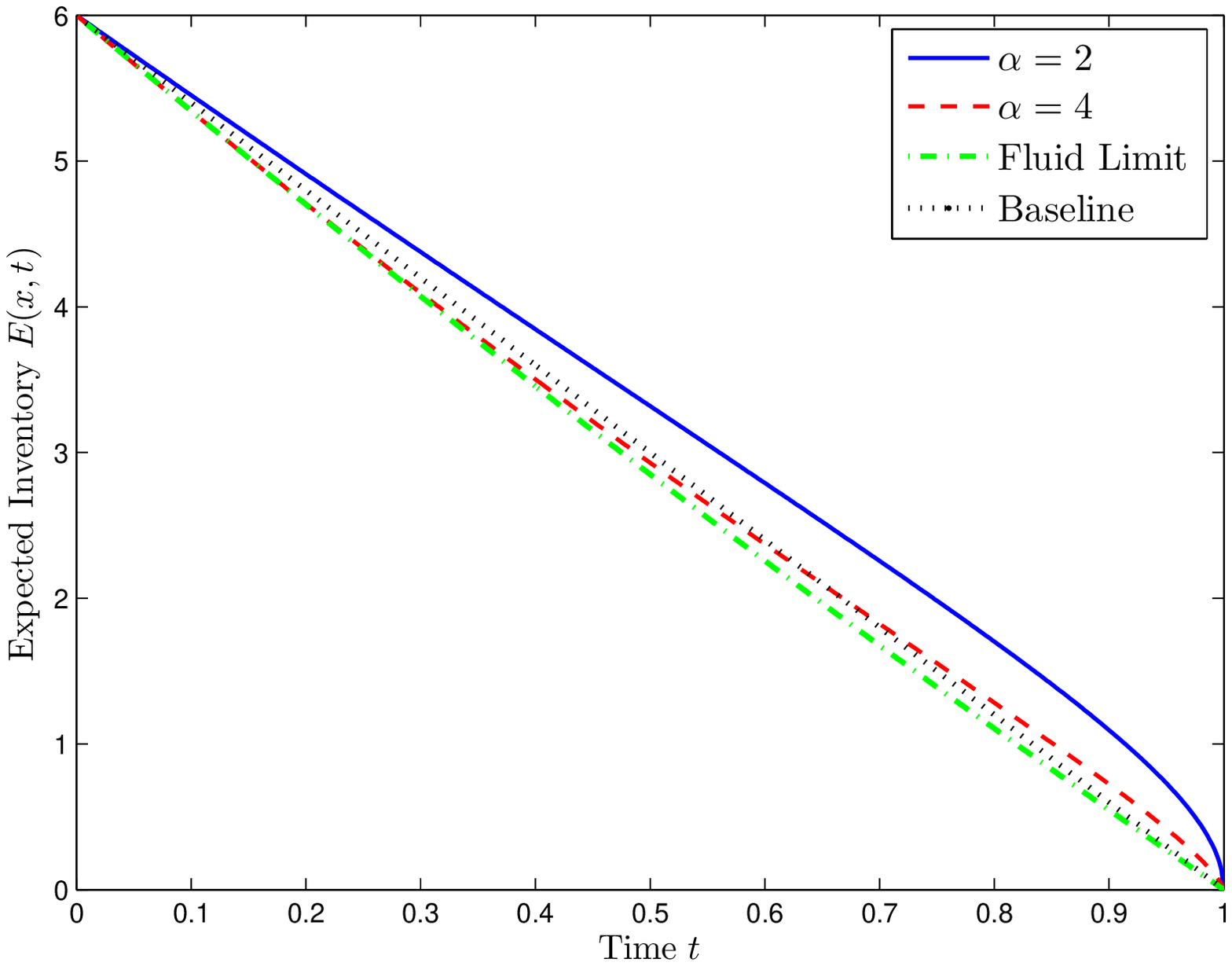}}
\end{minipage} & \begin{minipage}{0.5\textwidth}
\center{\includegraphics[width=3.35in]{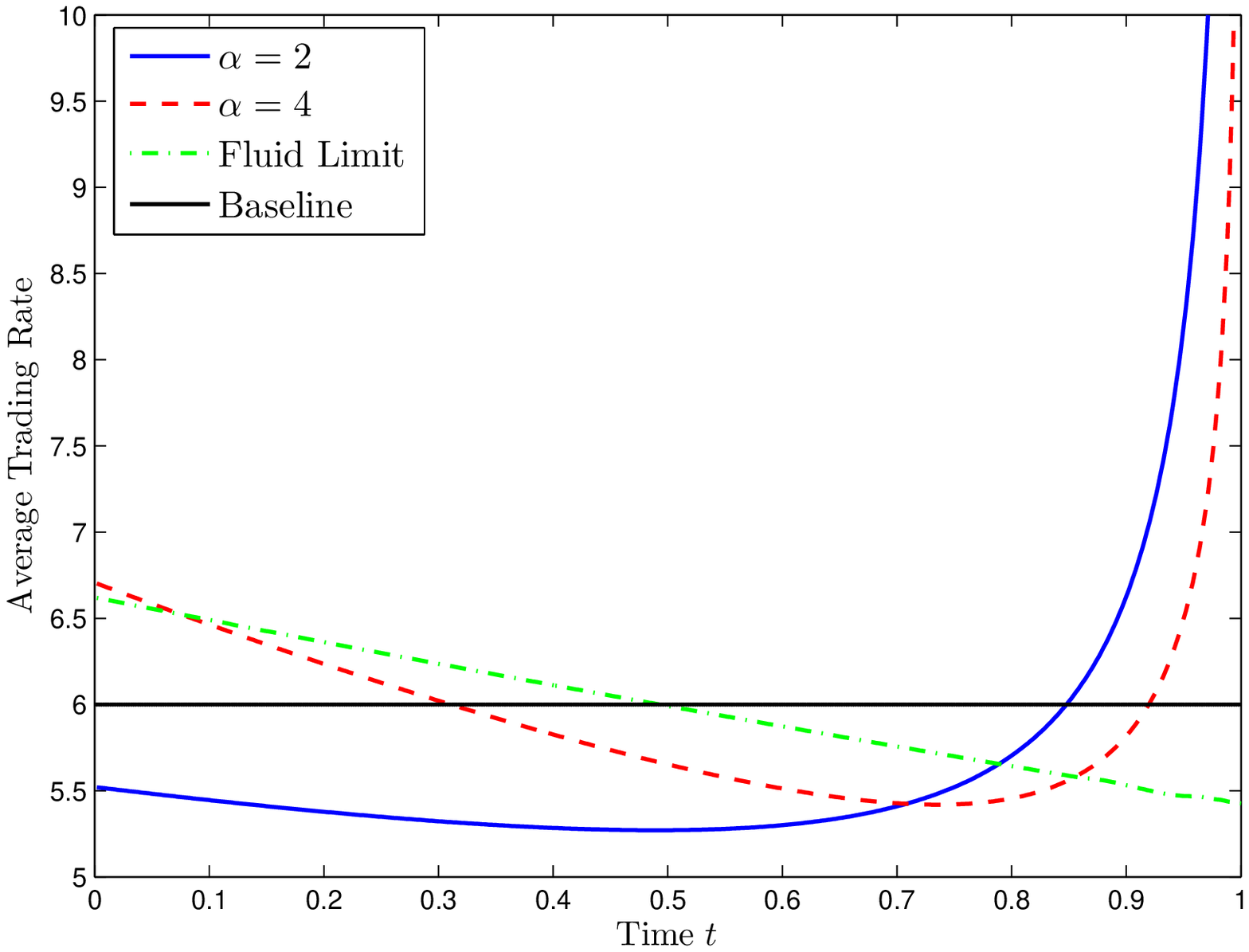}}
\end{minipage}\end{tabular}
\caption{Execution curves for different power-law books. We take $X_0 = 6$, $r=0.1$, $T=1$, and depth functions $\La(s) \in \{ s^{-2}, s^{-4}\}$.
\emph{Left panel:} average inventory $E(x,t)$ as a function of time $t$. \emph{Right panel:} average trading rate as a function of time.
\label{fig:power-trade-curve}}
\end{figure}

Thus,  whenever we have an explicit formula for the optimal spread $s^*(n,T)$, $E(x,t)$ is also available analytically (or more practically as a solution of an ODE). For the case of power-law order books, we obtain from Proposition \ref{Prop:1} that $$\La( s^*(k,T-t)) = \frac{\lambda^{-\frac{1}{\alpha-1}} (\alpha r c_k)^{\alpha/(\alpha-1)}}{1-\exp(-\alpha r (T-t))}.
$$
 We observe that depending on the parameter values (in particular initial inventory $k$ vis-a-vis order shape $\alpha$), $\La(s^*(k,T))$ could be smaller or bigger than $k/T$, i.e. the ordering between the initial trading rate and constant trading is ambiguous. At the other end,  as $t\to T$, the spread $s^*(k,t)$ goes to zero and consequently $\lim_{t\to T}d E(x,t)/dt = -\infty$.

 Figure \ref{fig:power-trade-curve} shows that as $\alpha$ increases, the shape of $t \mapsto E(x,t)$ changes substantially. In particular for large $\alpha$ (corresponding to ``thinner'' power laws), the execution curve has an S-shape, with trading rate high in the beginning and end of the time interval. On the other hand, for $\alpha$ small, the execution curve lies entirely above the baseline, i.e.~the limit order trader consistently executes slower. This occurs due to the two competing effects of trying to extract profit (which slows execution) and the time decay, i.e.~the need to make the deadline which speeds up trading. We observe that when the limit order book is deep (small $\alpha$), the profit effect dominates; this is a new phenomenon compared to most existing models, such as \cite{GueantLehalle11}.

Finally, as a comparison, Figure \ref{fig:power-trade-curve} also shows the deterministic case $t\mapsto X^{(0),x}_t$, see \eqref{eq:trade-curve-power-det}. In the latter case, $t \mapsto X^{(0),x}_t$ is strictly convex, which contrasts strongly with the pre-limit situation. As $\Delta \to 0$, execution risk vanishes and the discounting effect takes over, making the investor sell more in the beginning. Indeed, in the fluid limit the investor can smoothly drive $X^{(0),x}$ to zero, while for $\Delta >0$, $\PP(X^{*,x}_t > 0) > 0$ for any $t<T$, but $X^{*,x}_T=0$ since $dE(x,t)/dt|_{t=T-} = -\infty$.

In the exponential order book case with no discounting, we observe that the trading rate $\La(s^*(n,T-t))$ is \emph{independent} of the depth parameter $\kappa$. Moreover, numerical experiments suggest that $E(x,t)$ is (slightly) convex in $t$, i.e.~the trading rate is monotonically decreasing in time. This agrees with the classical results of \cite{Almgren00}.
Finally, we recall that Remark \ref{rem:exp-exec-curve} shows that in the fluid limit, the execution rate is constant over time, and $t \mapsto X^{(0),x}$ is linear. This occurs because $(X^{(0),x}_t,t)$ is the characteristic curve of the PDE given by \eqref{eq:exp-finh-pde}.
This phenomenon resembles the constant trading rate in Alfonsi et al.~\cite{schied07} who studied (continuous) trading through market orders only, with the LOB depth function driving the price impact mechanism. Again, we find a sharp dichotomy between the deterministic limit where $X^{(0),x}_T=0$ for $x \leq \lambda e^{-1} T$ and the stochastic version where $E(x,T) > 0$ strictly for all $x>0$.

\subsection{General Order Book Depth Functions}
Our basic setting can be readily extended to allow for more sophisticated or complex models. Below we review several such extensions; for ease of presentation we treat them in the stationary infinite-horizon setting.

Let us revisit the optimal execution problem for a generic order book depth function $\La(s)$.
In general, there are no closed-form expressions for $V(n)$ and the continuous fluid limit $v(x)$ becomes a useful analytic tool to understand the solution structure. In that regard, both Theorem \ref{thm:fir} and Corollary \ref{cor:fc} continue to hold under some reasonable assumptions on the intensity function $\Lambda$.

\begin{thm}\label{thm:general-lob}
Consider the optimal liquidation problem on infinite horizon with a general intensity depth function $\Lambda$.
Then the statements of Theorem~\ref{thm:fir} hold, and if we further
assume that the function $x \mapsto \Lambda(x)$ is decreasing and that
\begin{equation}\label{eq:cond-on-L}
\frac{\Lambda(x)\Lambda''(x)}{(\Lambda'(x))^2}<2, \quad \forall x \in \R_+,
\end{equation}
then both $V^{\Delta}$ and $v$ are concave, the corresponding controls $s^{(\Delta)}$ and $s^{(0)}$ are decreasing and the conclusion of  Corollary \ref{cor:fc} still holds.
\end{thm}

\begin{rem}
For power-law LOB's, condition \eqref{eq:cond-on-L} holds precisely when $\alpha > 1$, while for exponential LOB's it always holds. Both of these order books have decreasing intensity functions.
\end{rem}

\begin{proof}
The proof of Theorem~\ref{thm:fir} can be done without much change since we did not make use of any special properties of $\La(s)$ there. We will prove the stated concavity and monotonicity properties from which the statement of Corollary~\ref{cor:fc} follows immediately as before.

 By time-stationarity between trading dates the controls are constant and the dynamic programming principle until the first jump time for $V^\Delta(x)$ gives
$$
V^\Delta(x) = \sup_{s \ge 0} \int_0^\infty \frac{\Lambda(s)}{\Delta} e^{-(\Lambda(s) \Delta^{-1} +r)t} (s \Delta +V^\Delta(x-\Delta))\, dt = \sup_{s \ge 0} \frac{\Lambda(s)}{\Lambda(s)+ r \Delta} \left(s \Delta +V^\Delta(x-\Delta) \right).
$$
Differentiating the right-hand-side with respect to $s$, the first order condition for $s^* \equiv s^{(\Delta)}(x)$ is
 \begin{align*} r\Lambda'(s^*)( s^*  \Delta+V^\Delta(x-\Delta)) + \Lambda(s^*)(\Lambda(s^*)+ r \Delta) & = 0 \\ \Longleftrightarrow \quad rV^\Delta (x-\Delta) = -r s^* \Delta  - \frac{\La(s^*)}{\La'(s^*)}(\La(s^*)+ r \Delta) & := F(s^*).
\end{align*}
$V^\Delta$ is non-decreasing; therefore, if the derivative of $F$ is negative, then $s^{(\Delta)}(x)$ decreases in $x$. Explicitly,
\begin{align}\label{eq:discrete-s-dec}
F'(s^*) = -\left(r\Delta +\La(s^*)\right)\left[ 2 - \frac{\La \La''}{(\La')^2}(s^*)\right] < 0 \quad \Longleftrightarrow \quad 2 > \frac{\La(s^*) \La''(s^*)}{(\La'(s^*))^2}.
\end{align}
 Thus, \eqref{eq:cond-on-L} is sufficient for $x \mapsto s^*(x)$ to be decreasing.
Under this assumption and the assumption that $\Lambda$ is decreasing
we would have that
 $x \mapsto \frac{\Lambda(s^*(x))}{\Lambda(s^*(x))+ r \Delta}$ is increasing,
and as a result
 \begin{align*}
V^\Delta(x) - V^\Delta(x-\Delta) &= \frac{\Lambda(s^*(x))}{\Lambda(s^*(x))+ r \Delta} (s^*(x) \Delta +V^\Delta(x-\Delta)) \\ & \qquad \qquad - \frac{\Lambda(s^*(x-\Delta))}{\Lambda(s^*(x-\Delta))+ r \Delta} \left(s^*(x-\Delta) \Delta +V^\Delta(x-2\Delta) \right) \\
& \le \frac{\Lambda(s^*(x))}{\Lambda(s^*(x))+ r \Delta} \left\{s^*(x-\Delta) \Delta +V^\Delta(x-\Delta) - s^*(x-\Delta)\Delta - V^\Delta(x-2\Delta)\right\} \\
& \le V^\Delta(x-\Delta) - V^\Delta(x-2\Delta),
\end{align*}
so that $V^\Delta$ is concave.

Concavity of $V^{\Delta}$ on the other hand implies the concavity of $v$.
This is thanks to the first assertion of the theorem from which we know that $V^{\Delta}$ converges to $v$ uniformly on compacts.
Next we will show that $s$ is a decreasing function. The value function $v$ satisfies the first order PDE:
\[
\sup_{s \geq 0} \Lambda(s) (s-v')-rv=0, \quad v(0)=0.
\]
Optimizing over $s$ yields
\[
v'(x) = s^{(0)}(x)+\frac{\Lambda(s^{(0)}(x))}{\Lambda'(s^{(0)}(x))}.
\]
Hence $v$ is concave if and only if the right-hand-side above is a decreasing function of $x$.
However, it
follows from \eqref{eq:cond-on-L} that the function
\[
y \mapsto y + \frac{\Lambda(y)}{\Lambda'(y)}, \quad y \in \R_+,
\]
is increasing. As a result the concavity of $v$, which we have already shown, is equivalent to $x \mapsto s^{(0)}(x)$ decreasing. \end{proof}

\subsection{Regime Switching Market Liquidity}
Empirical evidence suggests that market liquidity is not constant, \cite{CarteaJaimungal10}. As a first step towards capturing more complex liquidity behavior, we consider a simple regime-switching model for market activity level in the power-law LOB's. More precisely, suppose that arrival rates are modulated by a two-state Markov chain $M$ with states $\{0,1\}$, 0 representing an active market and 1 representing a slow market. We will denote the transition rate from 0 to 1 by $\theta_0$, and the transition rate from 1 to 0 by $\theta_1$. Under regime 0, the arrival rates of the orders are
$\Lambda_0(s)=\lambda_0/s^{\alpha}$, and under regime 1, the arrival rates of the orders are $\Lambda_1(s)=\lambda_1/s^{\alpha}$. We will take $\lambda_0>\lambda_1$. We assume that $M$ is observed and known by market participants.

Denote by $U(n)$ (respectively $W(n)$) the infinite-horizon value function for an inventory of $n$ shares under the active (resp.~slow) market regime.
The value functions satisfy the following system of equations:
\[
\begin{split}
A_{\alpha} \lambda_0 [U(n)-U(n-1)]^{1-\alpha}-r U(n)+\theta_0 [W(n)-U(n)]&=0, \\
A_{\alpha} \lambda_1 [W(n)-W(n-1)]^{1-\alpha}-r W(n)+\theta_1 [U(n)-W(n)]&=0,
\end{split}
\]
with terminal condition $U(0)=W(0)=0$.
The continuous selling approximation of these functions, which we denote by $u$ and $w$ respectively satisfy the following system of ordinary differential equations:
\begin{align}\label{eq:cont-u-w} \left\{
\begin{aligned}
A_{\alpha} \lambda_0 u_{x}^{1-\alpha}-(r+\theta_0)u(x)+\theta_0 w(x)&=0,\\
A_{\alpha} \lambda_1 w_{x}^{1-\alpha}-(r+\theta_1)w(x)+\theta_1 u(x)&=0,
\end{aligned}\right.
\end{align}
with $u(0)=w(0)=0$.

\begin{prop}\label{prop:regime-switching-power}
The solutions to \eqref{eq:cont-u-w} are $u(x) = c^*_0 x^p$ and $w(x) = c_1^* x^p$  where $p=(\alpha-1)/\alpha$ and
\begin{align}\label{eq:regime-switching-c}
\left(\frac{\lambda_1}{r \alpha}\right)^{1/\alpha}<c^*_1<c^*_0<\left(\frac{\lambda_0}{r \alpha}\right)^{1/\alpha}.
\end{align}
\end{prop}
Observe that the bounds in \eqref{eq:regime-switching-c} correspond to the single-regime solutions given in \eqref{eq:toy-pde-sol} with $T=+\infty$.

\begin{proof}
We begin with an ansatz of $u(x)=c_0x^{p}$ and $w(x)=c_1 x^{p}$ with $p$ given in the statement of the proposition. Comparing with \eqref{eq:cont-u-w}, the coefficients $c_0$ and $c_1$ need to satisfy
\[
\begin{split}
A_{\alpha} \lambda_0 p^{1-\alpha}c_0^{1-\alpha}-(r+\theta_0)c_0+\theta_0 c_1&=0,\\
A_{\alpha} \lambda_1 p^{1-\alpha}c_1^{1-\alpha}-(r+\theta_1)c_1+\theta_1 c_0&=0.\\
\end{split}
\]
Re-writing as
 \begin{equation}\label{eq:twfn}
 \begin{split}
 c_0&=\frac{r+\theta_1}{\theta_1}c_1-\frac{\lambda_1}{\alpha \theta_1}c_1^{1-\alpha},\\
 c_1&=\frac{r+\theta_0}{\theta_0}c_0-\frac{\lambda_0}{\alpha \theta_0}c_0^{1-\alpha},
 \end{split}
 \end{equation}
it easily follows that this system of equations has a unique solution $(c^*_0,c^*_1)$. Indeed,
$c_0$ as a function of $c_1$ is strictly increasing and goes from $-\infty$ at $c_1=0$ to $\infty$ and $c_1=\infty$. Similarly, $c_1$ as a function of $c_0$ is also strictly increasing from $-\infty$ at $c_0=0$ to $\infty$ at $c_0=\infty$. It directly follows from these facts that these two curves intersect and do so at only one point. Moreover, the identity function $c_0=c_1$, intersects the first function \eqref{eq:twfn} first, and the second function in the same equation last. This proves the ordering in \eqref{eq:regime-switching-c} and concludes the proof.
\end{proof}

Proposition \ref{prop:regime-switching-power} shows that the asking spread will always be higher under the active market regime when the order book is deeper.

Figure \ref{fig:regime-switching} illustrates the impact of multiple liquidity regimes. We take $\la_0 = 1.5, \la_1 = 0.5$ so that trade intensity is tripled in the active regime. We plot $c_i^*$ as a function of $\theta_0 = \theta_1 = \theta$ for $r =0.1$ and $\alpha =2$. As $\theta \to 0$, we have $c_i^* \to \sqrt{ 0.5 \la_i r^{-1}}$, while as $\theta \to \infty$, $c_i^* \to \sqrt{\frac{1}{2 r} \frac{\la_0 + \la_1}{2}}$ the fast-switching limit.

\begin{figure}[ht]
\center{\includegraphics[height=3in]{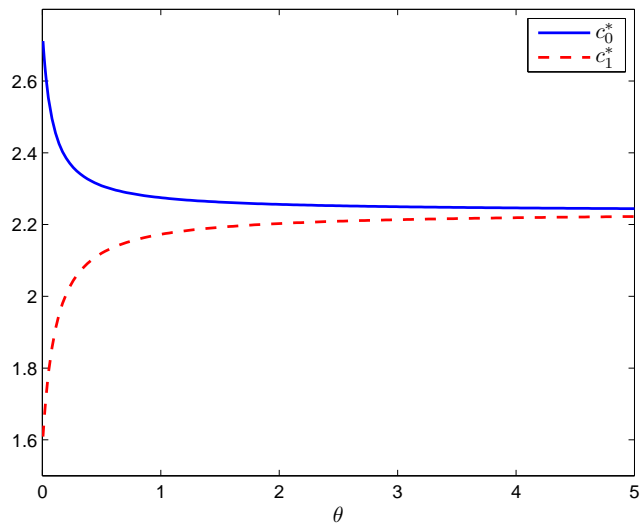}}
\caption{Regime switching model. We take $\alpha_0  = 2, \lambda_0 = 1.5, \lambda_1 = 0.5$ and $r=0.1$. The regime-switching rates are equal $\theta_0 = \theta_1 = \theta$.
 \label{fig:regime-switching}}
\end{figure}

\subsection{Two-Exchange Multi-scale Model}\label{sec:multi-scale}
Another possibility is to consider an investor trading on multiple venues. For example, suppose the investor can liquidate her holdings through two different exchanges, with each exchange possessing its own LOB. To distinguish the two exchanges,
we suppose that on exchange C(ontinuous) the orders are infinitesimally small, but on exchange L(arge) they are of \textit{large but finite size} relative to the total order size. More precisely, we assume that in the continuous limit, the exchange C orders are infinitesimal, but exchange L orders are of size $\delta$. If the remaining inventory is less than $x$ we assume that the next trade on exchange L will liquidate the entire $x$. In other words, actual trades on exchange L are of size $\min(\delta, x)$.

To keep the model tractable, we assume that each exchange has power-law depth with identical depth parameter $\alpha > 1$.
The resulting time-stationary value function $v(x)$ solves
\begin{align}\label{eq:inf-horizon-2-exchange}
\sup_{s_0 \geq 0} {\la}_0 \frac{s_0-v_x}{s_0^{\alpha}}+\sup_{s_1 \geq 0}\la_1 \frac{(\delta \wedge x) s_1-(v(x)-v( (x-\delta)_+))}{s_1^{\alpha}}-rv=0, \quad v(0)=0.
\end{align}
Plugging in the first-order optimizers leads to
\begin{align}\label{eq:v-delay-ode}
A_{\alpha} \lambda_0 v'(x)^{1-\alpha} + A_{\alpha} \lambda_1 (x \wedge\delta)^\alpha (v(x)-v((x-\delta)_{+}))^{1-\alpha} - r v = 0, \qquad v(0) = 0.
\end{align}

\begin{lemma}
There is a unique solution to \eqref{eq:v-delay-ode}. 
\end{lemma}

\begin{proof}
Equation \eqref{eq:v-delay-ode} is a first order nonlinear delay ODE and can be solved  by successive patching. Namely, first solve the ODE
\begin{equation}\label{eq:ODEfi}
A_{\alpha} \lambda_0 v_{(0)}'(x)^{1-\alpha} + A_{\alpha} \lambda_1 x^\alpha( v_{(0)}(x))^{1-\alpha} - r v_{(0)}= 0,
\end{equation}
with $v_{(0)}(0)=0$ on $[0,\delta]$. We then solve
\begin{align}\label{eq:ode-fi-2}
A_{\alpha} \lambda_0 v_{(1)}'(x)^{1-\alpha} + A_{\alpha} \lambda_1 \delta^\alpha( v_{(1)}(x) - v_{(0)}(x-\delta))^{1-\alpha} - r v_{( 1)}= 0
\end{align}
on $[\delta,2\delta]$ with initial condition $v_{(1)}(\delta) = v_{(0)}(\delta)$.  In \eqref{eq:ode-fi-2} we treat $v_{(0)}$ as a source term, observing that $v_{(0)}(x-\delta)$ with $x\in[\delta,2\delta]$ has already been computed before. Proceeding in this fashion, we finally set $v(x) = v_{(n)}(x)$ for $x \in [n\delta, (n+1)\delta]$ to recover the global solution. On each of the intervals $[n\delta, (n+1)\delta]$ the corresponding ODE has a locally Lipschitz driver so classical results give existence/uniqueness of solution $v_{(n)}$.
\end{proof}

\begin{rem}
Numerical computation of the solution of \eqref{eq:ODEfi} should be handled with care since $v_{(0)}'(0)=\infty$. We get around this singularity using the following observation:
For $x$ small enough, the benefit of large orders is negligible since the probability of getting a large order is very small. Therefore, close to zero, $v(x) \simeq v_0(x) = (\frac{ \la_0}{\alpha r})^{1/\alpha} x^{\frac{\alpha-1}{\alpha}}$ from \eqref{eq:toy-pde-sol}.

We also remark that the solution to \eqref{eq:ODEfi} is in general no longer concave, with concavity likely to fail around the knots $\delta, 2\delta,\ldots$, where the derivative $v'$ does not exist.
\end{rem}

Typically, trading intensity on the small-order exchange is several magnitudes larger than via the big trades (done through e.g.~a proprietary dark pool, see e.g., \cite{SchiedKlock11} where a single large dark pool trade liquidates the entire position), so $\la_0 \gg \la_1$. Fixing the time-scale as $\la_0 = 1$, we are therefore led to consider an asymptotic expansion in small $\la_1$. Formally, let $\la_1 = \bar{\la} \eps$ for $\eps$ small and consider a power series expansion in $\eps$,
\[
v(x) = v_0(x) + \eps v_1(x) + \eps^2 v_2(x) + \ldots,
\]
Plugging into \eqref{eq:v-delay-ode} and matching powers of $\eps$ we find that $v_0(x)$ solves the 1-exchange problem of \eqref{eq:toy-pde-sol}, so that $v_0(x) = (\frac{ \la_0}{\alpha r})^{1/\alpha} x^{\frac{\alpha-1}{\alpha}}$. Next,
\[
A_\alpha \alpha (1-\alpha) r x v_1'(x) + A_\alpha \bar{\la} (\delta \wedge x)^\alpha \left( \frac{\la_0}{\alpha r} \right)^{\frac{1-\alpha}{\alpha}} \bigl( x^{\frac{\alpha-1}{\alpha}} - (x-\delta)_+^{\frac{\alpha-1}{\alpha}} \bigr)^{1-\alpha} - r v_1(x) = 0.
\]
This is a first-order linear ODE with non-constant coefficients and therefore $v_1(x)$ can be expressed in closed-form using integrating factors as
\begin{align}
v_1(x) = \left\{ \begin{aligned} C_1 x^{2-\alpha^{-1}}, & \qquad x \le \delta, \\
 x^{-B_\alpha} \cdot \int_0^x C_2 y^{B_\alpha+\alpha-1} \bigl( y^{\frac{\alpha-1}{\alpha}} - (y-\delta)^{\frac{\alpha-1}{\alpha}}\bigr)^{1-\alpha} dy, & \qquad x > \delta, \end{aligned}\right.
\end{align}
with $C_1 = \frac{\bar{\la} \alpha r}{(1-\alpha)} (\frac{\la_0}{ \alpha r})^{(1-\alpha)/\alpha} \frac{1}{B_\alpha - 2+\alpha^{-1}}$, $C_2 = A_\alpha \bar{\la} \delta^\alpha B_\alpha \left( \frac{\lambda_0}{ \alpha r}\right)^{(1-\alpha)/\alpha}$ and $B_\alpha = \frac{\alpha^{\alpha-1}}{(\alpha-1)^\alpha}$. The latter integral only involves powers of $y$ and can be easily computed numerically. Similarly, the equations for higher-order terms are again first-order linear ODEs and so $v_2$, etc., can be written iteratively in closed-form.

\bibliographystyle{siam}
\bibliography{references}

\end{document}